\documentclass{article}

\usepackage{arxiv}
\usepackage[utf8]{inputenc} 
\usepackage[T1]{fontenc}    

\usepackage{hyperref}       

\usepackage{url}            

\usepackage{booktabs}       
\usepackage{amsfonts}       
\usepackage{nicefrac}       
\usepackage{microtype}      
\usepackage{fancyhdr}       
\usepackage{graphicx}       

\usepackage{textcomp}

\usepackage{amsmath}
\usepackage{amsthm}
\usepackage{xspace}

\newtheorem{definition}{Definition}
\newtheorem{theorem}{Theorem}
\newtheorem{corollary}{Corollary}

\newcommand{\FAB}{\textsc{FAB}\xspace}

\newcommand{\classP}{\textrm{P}}
\newcommand{\classNP}{\textrm{NP}}

\newcommand{\classPSPACE}{\textrm{PSPACE}}

\title{Optimizing for aggressive-style strategies in Flesh and Blood is NP-hard}

\author{
  Leonardo Gasparini Romão \\
  Institute for Technological Research (IPT) \\
  São Paulo, Brazil\\
  \texttt{leonardo.romao@ensino.ipt.br} \\
   \AND
  Samuel Plaça de Paula \\
  Institute of Computing, State University of Campinas \\
  Campinas, São Paulo, Brazil\\
  \texttt{samuel.paula@alumni.usp.br} \\
  \And
  Eduardo Takeo Ueda \\
  Institute for Technological Research (IPT) \\
  São Paulo, Brazil\\
  \texttt{eduardoueda@ipt.br} \\
}

\begin{document}
\maketitle

\begin{abstract}
    Flesh and Blood™ (FAB) is a trading card game that two players need to make a strategy to reduce the life points of their opponent to zero. The mechanics of the game present complex decision-making scenarios of resource management. Due the similarity of other card games, the strategy of the game have scenários that can turn an NP-problem. This paper presents a model of an aggressive, single-turn strategy as a combinatorial optimization problem, termed the FAB problem. Using mathematical modeling, we demonstrate its equivalence to a 0-1 Knapsack problem, establishing the FAB problem as NP-hard. Additionally, an Integer Linear Programming (ILP) formulation is proposed to tackle real-world instances of the problem. By establishing the computational hardness of optimizing even relatively simple strategies, our work highlights the combinatorial complexity of the game.
\end{abstract}


\section{Introduction}


Flesh and Blood™ (FAB) is a trading card game (TCG) developed by
Legend Story Studios® in 2019. Players assume the role of heroes with unique abilities and engage in battles where success relies on carefully balancing offense, defense, and resource generation.
The game's mechanics demand strategic planning, particularly for aggressive playstyles where players aim to maximize immediate damage output while managing limited resources.

In this paper, we work to explore the computational complexity of optimizing aggressive single-turn strategies in Flesh and Blood™. By modeling these strategies as a combinatorial optimization problem, termed the FAB problem, we establish its equivalence to the 0-1 Knapsack problem, proving that it is NP-hard. Furthermore, we present an Integer Linear Programming (ILP) formulation to address real-world instances of the problem. This study contributes to the field by advancing the understanding of computational challenges in modern card games and providing a foundation for developing heuristic and approximation algorithms tailored to strategic gameplay.

\subsection{Related works}

Before presenting our results, we very briefly summarize a few other findings that are related or similar to or own. Demaine investigated the computational complexity of the popular card game UNO~\cite{DEMAINE201451}. They define four different variants of the game, two of which being \textit{extensions} of UNO in the sense that they allow cards not present in traditional versions of the game. The authors investigate the hardness of determining the winning player from a given state of the match. This is shown to be NP-hard for two of the variants, while efficient algorithms are described for the other two.

Zhang explored the complexity of the digital card game Hearthstone, proving that in a format without randomness, the game is PSPACE-hard~\cite{Zhang2023}. We remind the reader that this implies NP-hardness while being possibly stronger since we could have $\classP \neq \classPSPACE$ even if it turns out that $\classP = \classNP$~\cite{Papadimitriou1994}.

Other games have been proven to be Turing-complete, such as Magic: The Gathering~\cite{Churchill2019}. Churchill et al. show that it is possible to set up a game state that goes on to simulate an Universal Turing Machine purely as a consequence of the game’s rules. Consequently, the game itself is Turing-complete, and specific problems related to it such as determining the winning player from a given state are undecidable rather than merely hard. This is not the first game to be shown to have this property, though to the best of our knowledge it is the only card game so far. We remark that a problem being undecidable in the general case does not necessarily mean that \textit{every} instance is impossible (or hard) to solve; the result is rather is a reflection of the potential complexity allowed by the system.

\subsection{Our contribution}
\label{sec:contribution}

We propose a mathematical programming model of Flesh and Blood, described in Section~\ref{sec:model}. Specifically, it is an Integer Linear Program (ILP), meaning it has linear restrictions and objective function, with variables restricted to being integer. Our model comprises a single-turn decision, with the restrictions enforcing the game rules. We provide a family of possible objective functions, and prove all of them to be NP-hard. This is done by proving that the simplest strategy allowed by the family of objective functions is already NP-hard, while the others are generalizations of it. All of this is described in Section~\ref{sec:hardness}.

Our model comprises a simple single-turn decision which only considers the player's own cards to decide how to play. However simple, since this is already NP-hard, it implies that any more robust models which extend it are also at least NP-hard.

Finally, we note that since we formulate the game as an ILP, the model itself is another product of this work. The ILP program immediately provides one possible way of solving instances of the problem, by running an appropriate solver. This is an aspect of our contribution that we do not explore here, leaving it as a possibility for future work. One could also consider Dynamic Programming as a possible way of deriving pseudo-polynomial algorithms, like the one for 0-1 Knapsack.

\section{The Game and Its Main Mechanics}

Here we give a simplified description of the game and its rules, only enough to motivate the model in the following Section~\ref{sec:model}. The main mechanics of Flesh and Blood revolve around some key concepts: hero selection, resource management, combat, and deck building. For more details, the reader can refer to our Appendix~\ref{rules} and the official documentation for the game~\cite{LegendStoryStudios2023}.

Players begin by placing their hero and equipment cards in the play area. The hero is placed at the center of the player's play area, with equipment distributed alongside. Then, each player adjusts their life point counter according to the chosen hero, which can vary for each character. The determination of the first player is done by a random method, such as flipping a coin, rolling a die, or according to the specific rules of the tournament or event they are participating in. After that, each player shuffles their deck and draws an initial hand of cards, usually composed of four cards, thus completing the battlefield preparation phase.

\begin{figure}[h]
    \centering
    \includegraphics[width=0.5\linewidth]{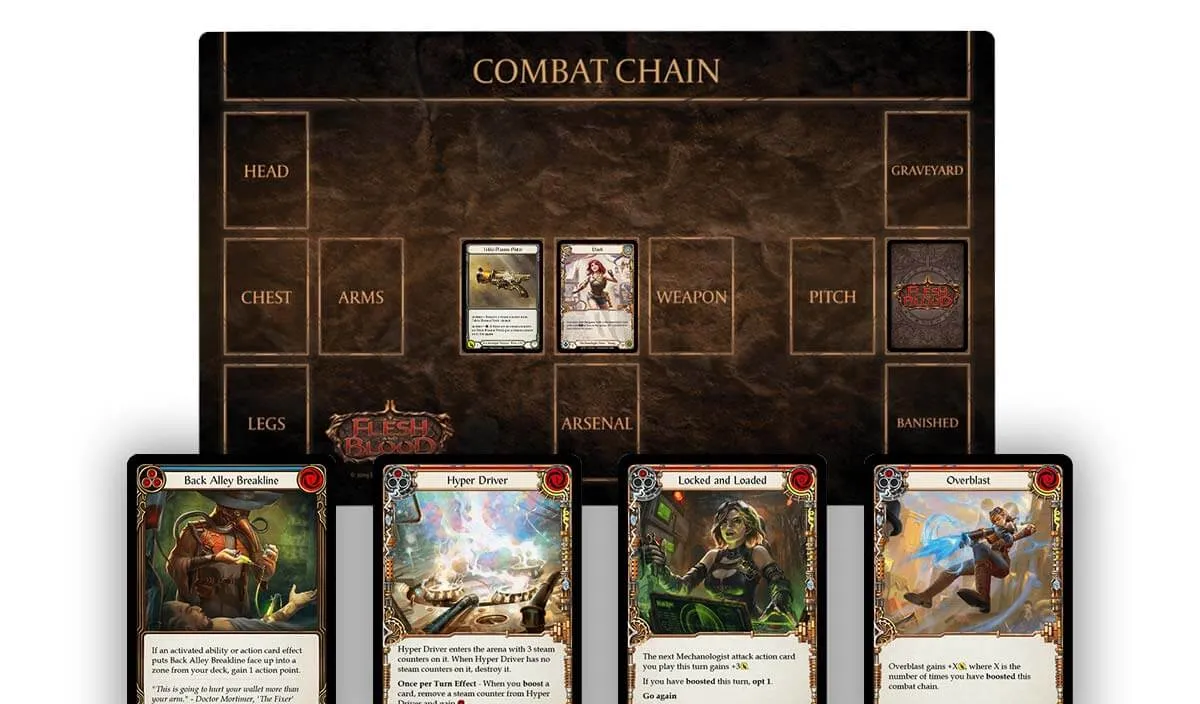}
    \caption{Game state during a turn in Flesh and Blood. Source: \cite{Zach2023}}
    \label{fig:game_state}
\end{figure}

\section{Modeling}
\label{sec:model}

As briefly described before in Section~\ref{sec:contribution}, we model the game as an Integer Linear Program. It represents the choices made by the player in one turn. The restrictions model the game rules, and the objective function will model the type of strategy.

In the next subsections, we describe how the problem of optimizing aggressive strategies in FAB is translated into a mathematical framework. We introduce the program variables and notations used to represent game elements, such as card attributes and player actions. Additionally, we explain how these elements are integrated into a combinatorial optimization problem, enabling a systematic analysis of decision-making within the game.

\subsection{Notations and Variables}

Here, we provide the definitions and notations used in our mathematical model in order to represent different elements of the game. For each element, we give a brief description of the element and provide our notation. In addition, some elements of the Mathematical Programming modeling will start to appear, though we reserve Section~\ref{sec:of} for the explicit models. In each subsection, we present a different aspect of the game that is a variable of the problem.

\subsubsection{Cards}

\paragraph{Definition:} In \FAB, a card \( c_i \) represents an action that a player can perform, such as attacking, defending, or using a special ability.
\paragraph{Notation:}
    \begin{itemize}
        \item \( \mathcal{C} = \{ c_1, c_2, \dots, c_n \} \) represents the set of cards available in a turn.
        \item Each card \( c_i \) has:
        \begin{itemize}
            \item \textbf{Attack value} \( a_i \): represents the damage the card can deal to the opponent.
            \item \textbf{Pitch cost} \( t_i \): represents the number of resources required to play the card.
            \item \textbf{Pitch resource} \( r_i \): represents the number of resources generated when using the card for pitching.
            \item \textbf{Defense value} \( d_i \): defensive capability of the card when kept in hand and used to defend.
        \end{itemize}
    \end{itemize}

\subsubsection{Pitch}

\paragraph{Definition:} Pitch is the resource used to pay the cost of playing cards in \FAB. Each card has a pitch value \( r_i \), which indicates how many resources it generates when used for pitching.
\paragraph{Notation:}
    \begin{itemize}
        \item \( T = \sum_{i \in S} t_i \): sum of the pitch costs of the cards played during a turn, where \( S \subseteq \mathcal{C} \) is the subset of cards played.
        \item \( R = \sum_{i \in P} r_i \): sum of the pitch resources generated by the cards used during a turn, where \( P \subseteq \mathcal{C} \) is the subset of cards used for pitching.
    \end{itemize}

\subsubsection{Attack Value}

\paragraph{Definition:} The attack value \( a_i \) of a card is the damage that the card can deal to the opponent when played.

\paragraph{Notation:}
    \begin{itemize}
        \item \( A = \sum_{i \in S} a_i \): sum of the attack values of the cards played during a turn.
    \end{itemize}

\subsubsection{Cards Kept for Defense}

\paragraph{Definition:} Cards not used for attack or pitch can be used for defense.
\paragraph{Notation:}
    \begin{itemize}
        \item \( D = \mathcal{C} \setminus (S \cup P) \): set of cards kept for defense.
        \item Total defense value available:
        \[
        D_{\text{total}} = \sum_{i \in D} d_i.
        \]
    \end{itemize}

To illustrate how cards work in the game, Figure~\ref{fig:card_example} shows how a card is printed with the items above. Next, we present what actions the player can do considering the variations of the type of cards in his hand. To formalize the player's actions in a game of Flesh and Blood, we create a mathematical formalization of the components. We assume that we are the first player to perform actions so that there is no intervention from the other player's actions.

\begin{figure}[h]
    \centering
    \includegraphics[width=0.5\linewidth]{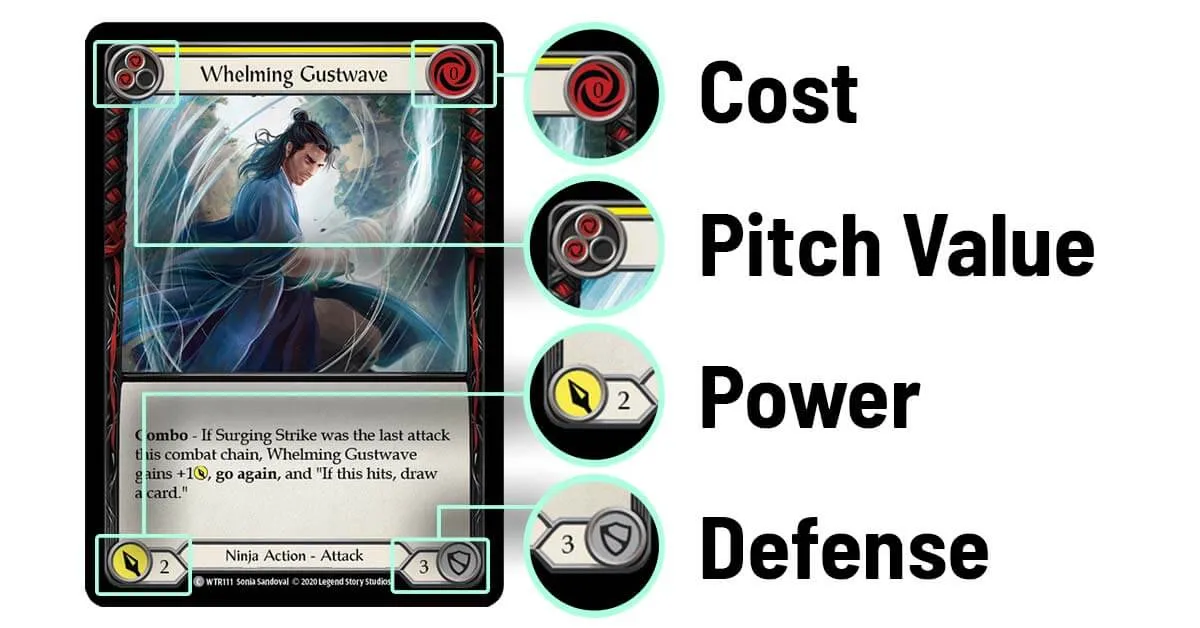}
    \caption{Example of a card in Flesh and Blood. Source: \cite{Zach2023}}
    \label{fig:card_example}
\end{figure}

\subsubsection{Player Actions}
\paragraph{Pitching to Generate Resources}:
    \begin{itemize}
        \item The player selects the set of cards \( P \subseteq \mathcal{C} \) to generate resources.
        \item The total resources generated are \( R = \sum_{i \in P} r_i \).
    \end{itemize}

\paragraph{Selection of Attack Cards}:
    \begin{itemize}
        \item The player selects the set of cards \( S \subseteq \mathcal{C} \) to attack.
        \item For each card \( c_i \in \mathcal{C} \):
        \begin{itemize}
            \item \( x_i \in \{0,1\} \): indicates whether the card \( c_i \) is played as an attack (\( x_i = 1 \)) or not (\( x_i = 0 \)).
            \item \( y_i \in \{0,1\} \): indicates whether the card \( c_i \) is used for pitch (\( y_i = 1 \)) or not (\( y_i = 0 \)).
            \item \( z_i \in \{0,1\} \): indicates whether the card \( c_i \) is kept for defense (\( z_i = 1 \)) or not (\( z_i = 0 \)).
        \end{itemize}
        \item Restricted to use only one time:
        \[
        x_i + y_i + z_i = 1, \quad \forall i \in \mathcal{C}.
        \]
    \end{itemize}

\paragraph{Executing Attacks}:
    \begin{itemize}
        \item The player plays the attack cards \( S \) and pays the cost in resources:
        \[
        \sum_{i \in S} t_i \leq R.
        \]
        \item The damage dealt to the opponent is:
        \[
        A = \sum_{i \in S} a_i.
        \]
    \end{itemize}

\subsection{Objective Functions and Hardness}

This section focuses on formalizing the FAB problem as a computational optimization challenge. It presents the primary objective functions for maximizing offensive output and balancing attack with defense, depending on the playstyle. Furthermore, it establishes the NP-hardness of the problem by demonstrating its equivalence to the 0-1 Knapsack problem. Variants of the FAB problem, such as FAB-AGGRO and FAB-MIDRANGE, are introduced to explore specific strategic contexts.

\label{sec:of}
\subsubsection{Research Question}

In the context of the aggressive playstyle in Flesh and Blood, how can we select a subset of available cards that maximizes the difference between the total damage dealt and the total pitch cost, respecting the constraints of available resources and cards in hand?

\begin{definition}
We define the \FAB problem as follows:

Given the set of available cards \( \mathcal{C} = \{ c_1, c_2, \dotsc, c_n \} \) and the total available pitch resources \( T \), select a subset \( A \subseteq \mathcal{C} \) maximizing the objective function, which is the difference between the total damage dealt and the penalty for the loss of defense when using cards for attack or pitch; that is:

\[
\text{Maximize } Z = \left( \sum_{i \in \mathcal{C}} a_i x_i \right) - \lambda \left( \sum_{i \in \mathcal{C}} d_i (x_i + y_i) \right)
\]

where $\lambda$  is the penalty factor associated with the loss of defensive capability.

\end{definition}

\begin{definition}
We define the \textsc{FAB-Aggro} problem as a restriction of the \FAB problem with \(\lambda = 0\), that is, optimizing purely for the maximum possible attack in a single turn.
\end{definition}

\begin{definition}
We define the \textsc{FAB-Midrange} problem as a restriction of the \FAB problem with \(\lambda = 1\), that is, optimizing the balance between attack in the current turn and defense in the next turn.
\end{definition}

Note that \textsc{FAB-Aggro} is the most restricted version, meaning the others are its generalizations. We will show that this problem is NP-hard because it models the 0-1 Knapsack problem, and therefore its variants are hard.

\section{Hardness Results}
\label{sec:hardness}

Having the above vocabulary defined, we may present our results regarding the computational complexity of the FAB problem. We prove its NP-hardness through a reduction from the 0-1 Knapsack problem to our \textsc{FAB-Aggro} problem. As noted above, this is the simplest form of the FAB problem. The analysis demonstrates how the core mechanics and decision-making challenges in Flesh and Blood can be mapped to a well-known combinatorial optimization problem. By formalizing this relationship, the subsection highlights the intrinsic difficulty of optimizing aggressive strategies in the game and provides a foundation for understanding the broader computational implications of strategic gameplay in similar contexts.

\begin{theorem}
The \textsc{FAB-Aggro} problem is NP-hard.
\end{theorem}

\begin{proof}
By reduction from the 0-1 Knapsack problem (\textsc{KP}). Consider a set of items \( I = \{ i_1, i_2, \dots, i_n \} \), where each item \( i_j \) has a benefit \( v_j \geq 0 \) when included in the knapsack and a weight \( w_j \geq 0 \) to be added to the knapsack. The knapsack has a maximum capacity \( W \geq 0 \). The \textsc{KP} has the following objective function:

\[
\text{Maximize } Z_{\text{KP}} = \sum_{i_j \in S} v_j
\]

Subject to the capacity constraint:

\[
\sum_{i_j \in S} w_j \leq W
\]

where \( S \subseteq I \) is the subset of items included in the knapsack. Also, the 0-1 Knapsack problem is know to be NP-Hard \cite{karp1972reducibility}

To reduce an instance of \textsc{KP} to the \textsc{FAB-Aggro} problem, we establish a correspondence between the items of \textsc{KP} and the cards of \textsc{FAB-Aggro}. For each item \( i_j \in I \), we create a card \( c_j \in \mathcal{C} \) with the following properties:

\begin{itemize}
    \item \textbf{Attack value}: \( a_j = v_j \).
    \item \textbf{Cost to play}: \( t_j = w_j \).
    \item \textbf{Pitch resource}: \( r_j = 0 \) (we assume no extra resources are generated).
    \item \textbf{Defense value}: \( d_j = 0 \) (this value is arbitrary, as defense is not considered in this variant of the problem).
\end{itemize}

In addition to the cards that correspond to knapsack items, we use special cards that represent the knapsack's capacity. These cards have no cost, value, or defense, and their sole purpose is to generate resources to allow the use of the cards.\footnote{There are item cards in Flesh and Blood; for this work, we use the "Energy Potion" card.}

In \textsc{KP}, the decision is whether to include an item \( i_j \) in the knapsack. In the \textsc{FAB-Aggro} problem, the decision is to use the card \( c_j \) for attack or pitch.

We establish the following correspondences:

\begin{itemize}
    \item If \( i_j \) is included in the knapsack (\( i_j \in S \)), then in \textsc{FAB-Aggro}:
    \begin{itemize}
        \item \( x_j = 1 \) (card used for attack).
        \item \( y_j = 0 \).
    \end{itemize}
    \item If \( i_j \) is not included in the knapsack (\( i_j \notin S \)), then in \textsc{FAB-Aggro}:
    \begin{itemize}
        \item \( x_j = 0 \), \( y_j = 0 \).
    \end{itemize}
\end{itemize}

\subsubsection*{Corresponding Objective Function}

Let \( x \) be a 0/1 vector indexed by the items; we denote \( w(x) = \sum_{i \in I} w_i x_i \) as the total weight of the chosen items. Similarly, let \( p(x) = \sum_{i \in I} t_i x_i \) be the total pitch cost of the chosen cards.

If \( (x^*, y^*) \) is an optimal solution of the \textsc{FAB-Aggro} instance, then

\begin{align*}
    \sum_{i \in I} x^*_i a_i
        &= \max_{r(x) \leq R} \sum_{i \in I} x_i a_i \\
        &= \max_{w(x) \leq W} \sum_{i \in I} x_i v_i \mbox{,}
\end{align*}

where the last equality holds because we constructed the \textsc{FAB-Aggro} instance such that \( R = W \) and, for all \( i \in I \), \( t_i = w_i \) and \( a_i = v_i \).

Therefore, \( x^* \) is an optimal solution to the \textsc{KP} instance.

\end{proof}

\begin{corollary}
\textsc{FAB-Midrange} and \FAB are NP-hard.
\end{corollary}

\begin{proof}
The \FAB problem generalizes \textsc{FAB-Aggro} because, by definition, \textsc{FAB-Aggro} is restricted to \(\lambda = 0\).

The \textsc{FAB-Midrange} problem generalizes \textsc{FAB-Aggro} because any instance of \textsc{FAB-Aggro} can be modeled by an instance of \textsc{FAB-Midrange} where the cards have no defense, and therefore the objective functions are identical.
\end{proof}

\section{Conclusion}

The reduction of the problem of manage damage in a single turn in Flesh and Blood (FaB) to a optimization problem of 0-1 Knapsack Problem provides significant insights into the computational complexity inherent in the game. By establishing a correspondence between the game’s elements and the parameters of these well-known optimization problems, we have demonstrated that the decision-making process in FaB is NP-hard. 


Recognizing the NP-hardness of the problem opens avenues for developing heuristic or approximation algorithms tailored to FaB, which can provide near-optimal solutions within reasonable computational times, aiding players in strategy formulation. For game designers, this analysis can guide the creation of new cards and mechanics, as understanding how certain attributes contribute to computational complexity can inform balanced game development.

Looking ahead, several questions and possible areas for future work emerge. Investigating heuristic approaches specific to FaB may reveal how effective these methods are compared to optimal strategies. Exploring approximation algorithms that guarantee solutions within a certain percentage of the optimal damage output raises considerations about acceptable trade-offs between computational time and solution quality in a real-time gaming context. 
Additionally, it looks similar to other games in how the structure of prizes can be obtained, suggesting if the basis of other card games can be reduced to an optimized problem of the Knapsack problem and its variations. Thus, we present possible solutions of the problem using integer linear programming, since we can use a solver to resolve the problems.



\appendix

\section{Basic Rules of the Game Flesh and Blood} \label{rules}

Flesh and Blood is a collectible card game where players assume the role of heroes battling in intense combat. Below are some of the basic concepts and rules of the game:

\subsection{Basic Concepts}

\begin{itemize}
    \item \textbf{Hero}: Each player chooses a hero to represent in combat. The hero determines the initial amount of life points and the class of cards the player can use.
    \item \textbf{Cards}: There are various types of cards, including attacks, defenses, and resource (pitch) cards. Each card has a resource cost, a pitch value, and may have additional abilities.
    \item \textbf{Resources}: To play cards, players need to generate resources, usually by discarding cards from their hands to use their pitch value.
    \item \textbf{Attack and Defense}: Players alternate turns attacking and defending. The goal is to reduce the opponent's hero's life points to zero before their own.
    \item \textbf{Go Again}: An important mechanic that allows the player to perform multiple actions in the same turn.
\end{itemize}

\subsection{Turn Structure}

A turn in Flesh and Blood is composed of several phases:

\begin{enumerate}
    \item \textbf{Action Phase}: The active player can play attack cards, activate abilities, and perform actions. Some actions may have the "Go Again" ability, allowing for subsequent actions in the same turn.
    \item \textbf{Reaction Phase}: The defending player can respond to attacks with defense and reaction cards.
    \item \textbf{End Phase}: The turn ends, and the active player prepares for the next turn, drawing cards up to the hand limit.
\end{enumerate}

For more details and advanced rules, consult the comprehensive rules and the game guide:

\begin{itemize}
    \item \textbf{Game Guide}: For a general guide on how to play Flesh and Blood, visit the TCGPlayer website \cite{Zach2023} 
    \item \textbf{Comprehensive Rules}: For a complete understanding of the game rules, access the Comprehensive Rules available at 
    \cite{LegendStoryStudios2023}
\end{itemize}

\bibliographystyle{unsrt} 
\bibliography{references}  

\end{document}